\documentclass[letterpaper, 10 pt, conference]{ieeeconf}

\IEEEoverridecommandlockouts                              

\usepackage[dvipsnames, table]{xcolor}
\usepackage[T1]{fontenc}
\usepackage{amsmath,amssymb}
\usepackage{siunitx}
\usepackage{subfigure}
\usepackage{multirow}
\usepackage[ruled,vlined]{algorithm2e}
\usepackage{bm}
\usepackage{url}
\usepackage{hyperref}
\usepackage{balance}
\usepackage{tikz}
\usepackage{psfrag}
\usepackage{mathtools}
\usepackage{comment}
\newcommand{\virg}[1]{``#1''}
\newcommand{\overbar}[1]{\mkern 1.5mu\overline{\mkern-1.5mu#1\mkern-1.5mu}\mkern 1.5mu}
\newtheorem{remark}{Remark}
\newtheorem{assumption}{Assumption}
\newtheorem{problem}{Problem}
\newtheorem{theorem}{Theorem}
\newtheorem{corollary}{Corollary}

\begin{document}
\title{\normalsize \textcolor{gray}{63rd IEEE Conference on Decision and Control, Accepted July 2024, Extended Version} \\[3mm] \Large Almost Global Trajectory Tracking for Quadrotors Using Thrust Direction Control on $\mathcal{S}^2$}
\author{Mirko Leomanni, Alberto Dionigi, Francesco Ferrante, Paolo Valigi, Gabriele Costante
\thanks{The authors are with the Department of Engineering, University of Perugia, 06125 Perugia, Italy. emails: \{mirko.leomanni, alberto.dionigi, francesco.ferrante, paolo.valigi, gabriele.costante\}@unipg.it.}}


\maketitle
\begin{abstract}
Many of the existing works on quadrotor control address the trajectory tracking problem by employing a cascade design in which the translational and rotational dynamics are stabilized by two separate controllers. The stability of the cascade is often proved by employing trajectory-based arguments, most notably, integral input-to-state stability. In this paper, we follow a different route and present a control law ensuring that a composite function constructed from the translational and rotational tracking errors is a Lyapunov function for the closed-loop cascade. In particular, starting from a generic control law for the double integrator, we develop a suitable attitude control extension, by leveraging a backstepping-like procedure. Using this construction, we provide an almost global stability certificate. The proposed design employs the unit sphere $\mathcal{S}^2$ to describe the rotational degrees of freedom required for position control. This enables a simpler controller tuning and an improved tracking performance with respect to previous global solutions. The new design is demonstrated via numerical simulations and on real-world experiments.
\end{abstract}

\section{Introduction}
Over the past few years, the scientific interest on quadrotor systems has steadily increased and the state of the art in quadrotor control has drastically
evolved \cite{mahony2012multirotor}. This comes not only from the versatility and broad applicability of these platforms \cite{longhi2017ubiquitous,longhi2018rfid}, but also from the fact that their dynamics evolve on a nonlinear
manifold, which makes the control problem nontrivial \cite{emran2018review}. Nowadays, a huge literature is available on such type of problems. Proportional-Integral-Derivative (PID) control is a commonly used solution and many PID variants have been proposed~\cite{Yu2015high,cao2015inner}. PID regulators provide adequate performance for simple setpoint stabilization. However, their application to more complex tasks requires a specialized design and a careful selection of the tuning parameters. To overcome these issues, more advanced control techniques have been explored. In \cite{mistler2001exact}, the trajectory tracking problem is addressed, by adopting a dynamic feedback linearization approach. Some extensions of this method are proposed in~\cite{lotufo2019control,leomanni2023robust} to enhance robustness against unmodeled dynamics and perturbations. To this purpose, backstepping \cite{cabecinhas2014nonlinear} and sliding mode \cite{lee2009feedback} control techniques have also been applied. Nonlinear model predictive control strategies are presented in \cite{torrente2021data,saviolo2022pitcn}. These are able to optimize performance while handling trajectory constraints. In most of the aforementioned works, stability results are only local.

Recently, a line of research has focused on global or almost global stabilization. Many of the contributions in this context address the control problem by employing a cascade design in which the translational and rotational dynamics of the quadrotor are stabilized by combining two separate controllers. The control strategies proposed in \cite{chaillet2014combining,roza2014class, naldi2016robust, invernizzi2018trajectory} ensure stability by requiring the position stabilizer to enforce an input-to-state stability (ISS) property with respect to small inputs. This is typically achieved by employing a nested saturation design (see. e.g., \cite{invernizzi2018trajectory}) which, however, may entail a poor transient performance \cite{marchand2005global}. To mitigate this drawback, \cite{invernizzi2019integral} proposes a hierarchical control design that establishes the stability of the cascade using a weaker property than ISS, namely, integral input-to-state stability (iISS) \cite{sontag1998comments}. This extends the class of allowable position stabilizers to more aggressive control laws. While ISS and iISS provide very useful qualitative information, they also display some drawbacks when quantitative statements are of interest. For instance, an explicit rate of convergence is often difficult to obtain \cite{grune2002input}.

\begin{figure}[t]
    \centering\vspace{2mm}
    \includegraphics[width=0.826\columnwidth]{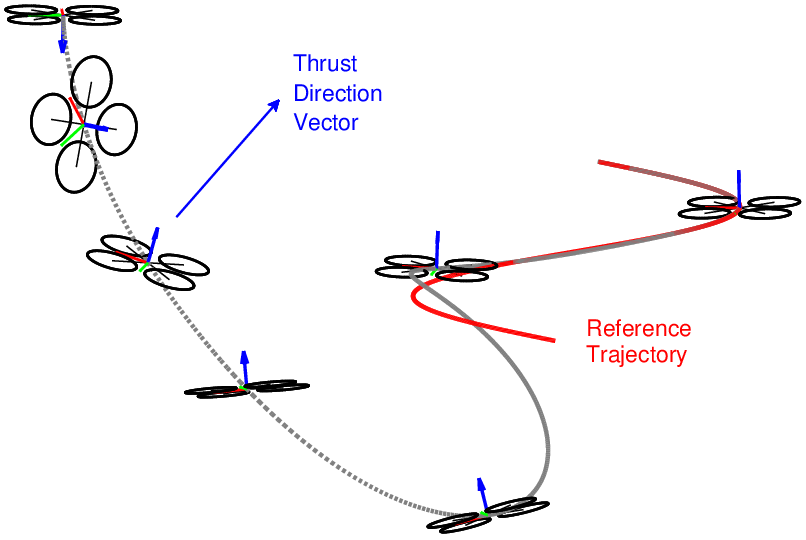}
    \caption{{Illustration of the proposed approach on a trajectory tracking maneuver in which the quadrtotor starts from an inverted flight configuration. The capability to perform such type of maneuvers is achieved by controlling the thrust direction vector directly on its configuration manifold $\mathcal{S}^2$}.}
    \label{fig:overview}
\end{figure}

Other recent works \cite{brescianini2018tilt, spitzer2020rotational, gamagedara2019geometric, kooijman2019trajectory} have noticed that, since the
translational dynamics of the quadrotor do not depend on its heading, it is natural to decouple the heading motion (i.e., the rotation about the thrust direction vector) from the tilt motion (i.e., the rotation of the thrust direction vector) for position control.
In particular, \cite{spitzer2020rotational} shows that separating the tilt and the heading configuration manifolds (the unit sphere $\mathcal{S}^2$ and unit circle $\mathcal{S}^1$, respectively) allows for improved position tracking capabilities with respect to using $\mathcal{SO}(3)$ rotations to compute the attitude error. In \cite{gamagedara2019geometric}, this idea is exploited to design a cascade control scheme for trajectory tracking. It is proven that the cascade enforces asymptotic stability in a local neighborhood of the equilibrium, whose size is dictated by complex relations among the tuning parameters. In \cite{kooijman2019trajectory}, a similar approach is pursued and an attitude control law achieving almost global stability on $\mathcal{S}^2\times\mathcal{S}^1$ is proposed. The control law is demonstrated in combination with a position controller. However, the stability analysis of the cascade is not carried out.

In this paper, we present a new control design methodology that allows a quadrotor to track a given position trajectory. The core idea behind our approach is to cast the tracking error dynamics in a form suitable for the application of recursive nonlinear control techniques, while operating directly in the configuration manifold of
the quadrotor. In particular, starting from a generic control law for the double integrator, we develop an attitude control extension defined on $\mathcal{S}^2$, by leveraging a backstepping-like procedure. Using this construction, we provide an almost global stability certificate, by means of a composite Lyapunov function. The proposed design overcomes the limitations of ISS-based formulations and extends the results in \cite{gamagedara2019geometric,kooijman2019trajectory} to cover the global stabilization of the tilt and position dynamics. The resulting controller has a simple structure and features a small number of tuning parameters. We demonstrate our design via numerical simulations and on real-world experiments.

The reminder of this work is organized as follows. Section~\ref{sec:2} formalizes the trajectory tracking problem, Section~\ref{sec:3} describes the proposed control desing, Section~\ref{sec:experiments} discusses the experimental results, and Section~\ref{sec:conclusion} draws the conclusions.

\emph{Notation}:
The set of real numbers is denoted by $\mathbb{R}$. The skew-symmetric cross product matrix representation of a vector $v=[v_x\,v_y\,v_z]^T\in\mathbb{R}^3$ is defined as
$ [v]_\times=
\begin{bsmallmatrix}
0&-v_z&v_y\\
v_z&0&-v_x\\
-v_y&v_x&0
\end{bsmallmatrix},
$
while the 2-norm of $v$ is denoted by $\|v\|$.
For convenience, we define the unit vector $\zeta=[0\; 0\; 1]^T.$ The identity matrix is denoted by $I$.

\section{Problem Formulation}\label{sec:2}

Two coordinate systems are employed in this work, namely, the inertial and body-fixed frames. The inertial frame originates at a fixed point in world coordinates. Its $Z$ axis points in the direction opposite to gravity, while the $X$ and $Y$ axes complete a right-handed triad. The body-fixed frame has its origin on the quadrotor and its axes are aligned to the principal axes of inertia of the vehicle. In particular, the body-fixed $Z$ axis is aligned with the thrust direction vector.

The quadrotor translational dynamics and its attitude kinematics are described by the model
\begin{eqnarray}
\ddot{p}&=& R^T \zeta f - \zeta g \label{sysmodel1}\\[1mm]
\dot{R}&=&-[\omega]_\times R\, \;, \label{sysmodel2}
\end{eqnarray}
where $p\in\mathbb{R}^3$ is the position vector of the center of mass of the quadrotor, expressed in inertial coordinates, $R\in\mathcal{SO}(3)$ is the attitude matrix which transforms a vector in inertial coordinates into a vector in body-fixed coordinates, $\omega\in\mathbb{R}^3$ is the angular velocity of the body-fixed frame relative to the inertial frame, expressed in body-fixed coordinates, $f\in\mathbb{R}$ is the vehicle acceleration along the thrust direction vector, and  $g=9.8$ m/s$^2$ is the standard gravity parameter.

In \eqref{sysmodel1}-\eqref{sysmodel2}, we treat $f$ and $\omega$ as control inputs, and neglect the angular velocity dynamics. This is done in order to simplify the control design problem. Indeed, the stabilization of the angular velocity dynamics can be tackled by applying a standard recursive design technique (e.g., backstepping \cite{kanellakopoulos1992toolkit}), once a suitable control law and a corresponding Lyapunov function are established for system \eqref{sysmodel1}-\eqref{sysmodel2}. Finding this pair is the focus of our discussion. In particular, we consider the following control problem.
\begin{problem}\label{pb1}
Find a feedback control law and a Lyapunov function ensuring that the output $p$ of system \eqref{sysmodel1}-\eqref{sysmodel2} asymptotically tracks a smooth time-varying reference $p_r\in \mathcal{P}$, where $\mathcal{P}\subset\mathbb{R}^3$ is a compact set.
\end{problem}

In order to solve Problem \ref{pb1}, we introduce the following assumption.
\begin{assumption}\label{A1}
Let $\xi_1\in\mathbb{R}^3$, $\xi_2\in\mathbb{R}^3$ and
\begin{equation}\label{distv}
d=\ddot{p}_r + \zeta g.
\end{equation}
There exist a positive constant $\epsilon$ and a differentiable control law $u_\xi(\xi_1,\xi_2,d)$ satisfying $\|u_\xi(\xi_1,\xi_2,d)\|\geq \epsilon$, such that the origin is a globally asymptotically stable equilibrium point for system
\begin{equation}\label{sysred}
\begin{array}{c c l}
\dot{\xi}_1&=& \xi_2 \\
\dot{\xi}_2&=& u_\xi(\xi_1,\xi_2,d) - d \, .
\end{array}
\end{equation}
Moreover, system \eqref{sysred} admits a differentiable Lyapunov function $V_\xi(\xi_1,\xi_2)$.
\end{assumption}

\begin{remark}
A necessary condition for the existence of a control law satisfying Assumption \ref{A1} is that $\|d\|\geq d_m$ for all times $t$, where $d_m$ is a positive constant (see, e.g., \cite{naldi2016robust}). Possible design options can be found in \cite{teel1992global} and \cite{forni2010family}.
\end{remark}

\section{Main Result}\label{sec:3}
We solve Problem \ref{pb1} in two steps. First, a convenient representation of the tracking error dynamics is derived. Then, the controller design is addressed by building on this representation.
\subsection{Tracking Error Dynamics}
The position and velocity tracking errors are described by the error variables
\begin{equation}\label{cchange}
\begin{array}{c c l}
x_1&=& p-p_r\\[1mm]
x_2&=& \dot{p}-\dot{p}_r \, .
\end{array}
\end{equation}
By differentiating \eqref{cchange} with respect to time and using \eqref{sysmodel1},\eqref{distv}, one obtains
\begin{equation}\label{reds}
\begin{array}{c c l}
\dot{x}_1&=& x_2\\[1mm]
\dot{x}_2&=& R^T \zeta f - d \, .
\end{array}
\end{equation}
Next, we define the \emph{desired thrust direction vector, expressed in the body-fixed frame}, as follows
\begin{equation}\label{auv}
x_3=R \frac{v}{\|v\|} \, ,
\end{equation}
where $v\in\mathbb{R}^3\setminus \{0\}$. Notice that, in this setting, $\zeta^T x_3$ is the cosine of the angle between the actual thrust direction vector and the desired one. By using, \eqref{auv}, one can rewrite the second equation in \eqref{reds} as
\begin{equation}\label{backstep}
\dot{x}_2=  \frac{v}{\|v\|} f- d +R^T(\zeta-x_3) f \, .
\end{equation}
From \eqref{sysmodel2} and \eqref{auv}, we get the following differential constraint
\begin{equation}\label{uvdyn}
\dot{x}_3=[x_3]_\times (\omega-R\omega_v) \, ,
\end{equation}
where
\begin{equation}\label{omegav}
\omega_v=\frac{[v]_\times}{\|v \|^2}\,\dot{v} \, .
\end{equation}
Additional details about the derivation of \eqref{uvdyn}-\eqref{omegav} are provided in Appendix \ref{appA1}.

By augmenting the first equation in \eqref{reds} with \eqref{backstep} and \eqref{uvdyn}, the following dynamic model is obtained
\begin{equation}\label{targetdyn}
\begin{array}{c c l}
\dot{x}_1&=& x_2\\[3mm]
\dot{x}_2&=&  \dfrac{v}{\|v\|}f- d + R^T(\zeta-x_3) f\\[4mm]
\dot{x}_3&=&[x_3]_\times (\omega-R\omega_v) \, .
\end{array}
\end{equation}
The vector $v$ in \eqref{targetdyn} plays the role of a \virg{virtual control}. This is employed next for control design.
\begin{remark}
System \eqref{targetdyn} is obtained by applying a backstepping-like technique to system \eqref{sysmodel2},\eqref{reds}. The adopted procedure takes advantage of the fact that $x_3$ evolves on the configuration manifold $\mathcal{S}^2$.
\end{remark}

\subsection{Control Design}
Let $x=[x_1^T\;x_2^T\;x_3^T]^T$. The key observation that we employ for control design is that Problem~\ref{pb1} can be recast as that of finding a \virg{virtual control} $v$ and a feedback control law which render $x=\overbar{x}$, with $\overbar{x}=[0\,0\,\zeta^T]^T$, an asymptotically stable equilibrium point for system \eqref{targetdyn}. We wish to determine a Lyapunov function $V$ for the resulting closed-loop system. To this aim, let us set
\begin{equation}\label{vdef}
v=u_\xi(x_1,x_2,d)
\end{equation}
in \eqref{auv} and \eqref{omegav}-\eqref{targetdyn}. For system \eqref{targetdyn}-\eqref{vdef}, we propose the continuous control law
\begin{equation}\label{claw}
\setlength{\arraycolsep}{3pt}
\begin{array}{c c l}
f&=& \| u_\xi(x_1,x_2,d) \|  \\[3mm]
(I-\zeta\zeta^T)\omega&=&  (I-\zeta\zeta^T)\big(R\omega_v \!+ [\zeta]_\times (\kappa_1 x_3+\beta)\big),
\end{array}
\end{equation}
where
\begin{equation} \label{betadef}
\setlength{\arraycolsep}{3pt}
\begin{array}{lll}
\beta&=& k_2(1+\zeta^T x_3)(\zeta^T\lambda)x_3-\dfrac{k_2 (1+\zeta^T x_3)^2 c}{1-\zeta^T x_3+c}\lambda \\[3mm]
& &-\dfrac{k_2(1+\zeta^T x_3)(x_3^T(I-\zeta \zeta^T)\lambda)}{1-\zeta^T x_3+c}x_3 ,
\end{array}
\end{equation}
\begin{equation} \label{lambdadef}
\lambda=\|u_\xi(x_1,x_2,d)\|R \left(\dfrac{\partial V_\xi(x_1,x_2)}{\partial x_2}\right)^{\!T} \!,
\end{equation}
$\kappa_1 \geq k_1$ is a smooth tuning function, and $k_1$, $k_2$, $c$ are scalar positive constants.

The next result establishes that the proposed control law provides an almost global solution to Problem~\ref{pb1}.

\begin{theorem}\label{th1}
Let Assumption \ref{A1} be satisfied. Then, the set $\mathcal{A}=\{(x,R):\,x=\overbar{x},\, R\in\mathcal{SO}(3)\}$ of system \eqref{sysmodel2},\eqref{targetdyn}-\eqref{lambdadef} is asymptotically stable, with region of attraction containing
\begin{equation}\label{dattract}
\mathcal{D}=\left\{ (x,R)\in  \mathbb{R}^9 \times \mathcal{SO}(3) :\; x_3\in\mathcal{S}^2\setminus\{-\zeta\}\right\} .
\end{equation}
Moreover, the composite function
\begin{equation}\label{compositeLyap}
V=V_\xi( x_1,  x_2)+\frac{1-\zeta^T x_3}{2 k_2(1+\zeta^T x_3)}
\end{equation}
is a Lyapunov function for system \eqref{sysmodel2},\eqref{targetdyn}-\eqref{lambdadef}.
\end{theorem}
\begin{proof}
Clearly, $V$ is a positive definite function with respect to the set $\mathcal{A}$ on the set $\mathcal{D}$ (notice that $\frac{1}{2}(x_3-\zeta)^T (x_3-\zeta)=1-\zeta^T x_3$), and it is continuously differentiable on $\mathcal{D}$. In particular, for all $(x,R)\in\mathcal{D}$, one has that
\begin{equation}\label{proofdvf}
\dot{V}\!=\!\dot{V}_\xi^*(x_1,x_2)+\lambda^T (\zeta-x_3)-\dfrac{\zeta^T[x_3]_\times [\zeta]_\times }{k_2(1+\zeta^T x_3)^2}(\kappa_1 x_3+\beta) ,
\end{equation}
where
\begin{equation}\label{dVxist}
\begin{split}
\dot{V}_\xi^*(x_1,x_2)=&\dfrac{\partial V_\xi(x_1,x_2)}{\partial x_1}x_2 \\&+ \dfrac{\partial V_\xi(x_1,x_2)}{\partial x_2}\left[ u_\xi(x_1,x_2,d)-d\right],
\end{split}
\end{equation}
$\lambda$ is given by \eqref{lambdadef}, and we used the identity $\zeta^T [x_3]_\times \omega=\zeta^T [x_3]_\times (I-\zeta \zeta^T)\omega$.
To continue the proof, we exploit the following relationships
\begin{equation}\label{skwid}
\setlength{\arraycolsep}{2pt}
\begin{array}{l l l}
\lambda^T(\zeta-x_3)&=&(1-\zeta^Tx_3)\zeta^T\lambda-x_3^T (I-\zeta \zeta^T)\lambda\\[2mm]
\zeta^T [x_3]_\times [\zeta]_\times&=&x_3^T(I-\zeta \zeta^T)\\[2mm]
x_3^T(I-\zeta \zeta^T)x_3&=&(1+\zeta^Tx_3)(1-\zeta^Tx_3) \, .
\end{array}
\end{equation}
Substituting \eqref{betadef} into \eqref{proofdvf} and using \eqref{skwid}, we get
\begin{equation}\label{proofdv}
\dot{V}=\dot{V}_\xi^*(x_1,x_2)-\dfrac{\kappa_1(1-\zeta^T x_3)}{k_2(1+\zeta^T x_3)} \,.
\end{equation}
By using Assumption \ref{A1}, it turns out that $\dot{V}_\xi^*(x_1,x_2)$ is nonpositive. Likewise, the second term on the right-hand side of \eqref{proofdv} is nonpositive on $\mathcal{D}$. Therefore, it can be shown that $\mathcal{D}$ is forward invariant for system  \eqref{sysmodel2},\eqref{targetdyn}-\eqref{lambdadef}\footnote{This property can be easily proved by contradiction, by relying on the fact that $V$ grows unbounded when approaching the boundary of $\mathcal{D}$.}. Moreover, by Assumption \ref{A1} and \eqref{compositeLyap}, it can be shown that there exists a class $\mathcal{K}_\infty$ function $\rho$ such that for all $(x,R)\in\mathcal{D}$, $V=V(x)\geq \rho(\|x-\bar{x}\|) $. This, in turn, implies that any trajectory starting in $\mathcal{D}$ is bounded. Then, by invoking LaSalle's invarance principle, it follows that any solution from $\mathcal{D}$ converges to the largest invariant set $\mathcal{M}$ contained in $\Omega$, where
\begin{equation}
\Omega=\left\{(x,R)\in \mathcal{D} : \, \dot{V}_\xi^*(x_1,x_2)=0,\, x_3=\zeta\right\} .
\end{equation}
By Assumption \ref{A1}, the origin is the only invariant set in which $\dot{V}_\xi^*(\xi_1,\xi_2)=0$, along the trajectories of system \eqref{sysred}. Then, from \eqref{auv} and the structure of \eqref{targetdyn}, it can be concluded that $\mathcal{A}\supset\mathcal{M}$. Being $\mathcal{A}$ invariant, the theorem statement is proven.
\end{proof}

\begin{remark}
In \eqref{claw}, the angular velocity of the quadrotor about the thrust direction vector, i.e., $\zeta^T \omega$, is not specified. This feature allows one to tackle the heading control problem separately from Problem \ref{pb1}, see, e.g., \cite{kooijman2019trajectory}.
\end{remark}

The next result stems directly from Theorem \ref{th1}.

\begin{corollary}\label{cor1}
Assume that
\begin{equation}\label{Vxiexp}
\dot{V}_\xi^*(\xi_1,\xi_2)\leq -\alpha {V}_\xi(\xi_1,\xi_2)
\end{equation}
for some $\alpha>0$, i.e., that $u_\xi(\xi_1,\xi_2,d)$ exponentially stabilizes the origin of system \eqref{sysred}. Then, the control law \eqref{claw}-\eqref{lambdadef} exponentially stabilizes the set $\mathcal{A}$.
\end{corollary}
\begin{proof}
Let $V_a={V}_\xi(x_1,x_2)$ and $V_b=\frac{1-\zeta^T x_3}{2 k_2(1+\zeta^T x_3)}$, so that $V=V_a+V_b$ according to \eqref{compositeLyap}. By applying these definitions to \eqref{proofdv} and using \eqref{Vxiexp}, one obtains
\begin{equation}\label{vbound}
\dot{V}= -\alpha V_a -2\kappa_1 V_b \leq -\min\{\alpha,\,2 k_1\} V \, .
\end{equation}
The inequality \eqref{vbound} indicates that the set $\mathcal{A}$ is exponentially stable, thereby concluding the proof.
\end{proof}

Notice that Corollary \ref{cor1} can be used to determine an explicit estimate of the rate of convergence of the closed-loop system trajectories.

The control law \eqref{claw}-\eqref{lambdadef} has a simple structure. In particular, the angular velocity control input $\omega$ includes a reference term $R\omega_v$ (obtained from \eqref{omegav} and \eqref{vdef}), a proportional term $\kappa_1 x_3$, and an additional contribution $\beta$. The latter is influenced by the position and velocity errors $(x_1,x_2)$, as well as by the choice of $V_\xi(x_1,x_2)$, see \eqref{betadef}-\eqref{lambdadef}. The overall construction is such that the attitude control gain increases with increasing translational errors, which provides an effective mechanism to improve the transient performance. The role of the constant $k_2$ in \eqref{betadef} is to scale the correction term $\beta$. This corresponds to scaling $V_\xi(x_1,x_2)$, and has clearly no impact on closed-loop stability. The parameter $c$ ensures the continuity of the control law.

\section{Experiments}\label{sec:experiments}
In this section, the proposed design is demonstrated via numerical simulations and real-world experiments. In order to deploy the control scheme \eqref{claw}-\eqref{lambdadef}, a position control law $u_\xi(\xi_1,\xi_2,d)$ and a corresponding Lyapunov function $V_\xi(\xi_1,\xi_2)$ must be specified. For simplicity, herein we consider the following state feedback plus feedforward controller
\begin{equation}\label{auxcon}
u_\xi(\xi_1,\xi_2,d)=-K \begin{bmatrix}\xi_1^T & \xi_2^T\end{bmatrix}^T + d \, ,
\end{equation}
where $K$ is a stabilizing gain. The control law \eqref{auxcon} may not satisfy the requirement $\|u_\xi(\xi_1,\xi_2,d)\| \geq \epsilon$ put forth by Assumption \ref{A1}. In our experiments, such a condition is verified a posteriori\footnote{The set of initial conditions such that $\|u_\xi(x_1,x_2,d)\|$ vanishes for some $t$, along the trajectories of the closed-loop system \eqref{sysmodel2},\eqref{targetdyn}-\eqref{lambdadef}, is expected to have a very small measure.}. The Lyapunov function $V_\xi(\xi_1,\xi_2)$ is specified as follows
\begin{equation}\label{auxV}
V_\xi(\xi_1,\xi_2)=\begin{bmatrix}\xi_1^T & \xi_2^T\end{bmatrix}P\begin{bmatrix}\xi_1^T & \xi_2^T\end{bmatrix}^T ,
\end{equation}
where the positive definite matrix $P$ is obtained by solving the Lyapunov equation
\begin{equation}\label{posV}
(A-BK)^T P+ P(A-BK) + I=0 \, ,
\end{equation}
being $A$ and $B$ the state and input matrices of the double integrator system \eqref{sysred}, respectively.

Based on laboratory tests with our quadrotor platform, the gain matrix $K$ is chosen as
\begin{equation}\label{kmat}
K=
\begin{bmatrix}
4 & 0 & 0 & 2 & 0 & 0\\
0 & 4 & 0 & 0 & 2 & 0\\
0 & 0 & 4.5 & 0 & 0 & 3\\
\end{bmatrix} \, .
\end{equation}
 Similarly to what is done in  \cite{kooijman2019trajectory}, the tuning function is $\kappa_1$ in \eqref{claw} is selected as follows
\begin{equation}\label{kappa1ch}
\kappa_1=\left\{
\begin{array}{c l}
k_1 & \text{if}\; \zeta^T x_3\geq 0 \\[2mm]
\dfrac{k_1}{\sqrt{1-(\zeta^T x_3)^2}} & \text{otherwise} \, .
\end{array}
\right.
\end{equation}
This choice is known to improve the convergence speed when the quadrotor starts from an upside down configuration \cite{spitzer2020rotational}.

The results discussed next are obtained by using \eqref{claw}-\eqref{lambdadef} and  \eqref{auxcon}-\eqref{kappa1ch}, along with ${k_1=1.5}$, ${k_2=0.05}$, ${c=0.1}$ and $\zeta^T\omega=0$. For digital implementation, the control law is discretized with a sampling frequency of $100$ Hz.

\subsection{Numerical Simulations}
A numerical simulation campaign is carried out on model \eqref{sysmodel1}-\eqref{sysmodel2} to verify the obtained theoretical results. In these simulations, the reference trajectory is specified as follows
\begin{equation}\label{rtr}
p_r(t)=\begin{bmatrix}0.38\, t & 0.6\sin\left(\dfrac{2\pi t}{10}\right) & 1\end{bmatrix}^T \; \text{m} \,.
\end{equation}

A first set of simulations is performed to evaluate the capability of the quadrotor to track the trajectory \eqref{rtr}, starting from different initial conditions. In particular, we set $p_r(0)=[x_0\;y_0\;z_0]^T$, $\dot{p}(0)=0$ and $R(0)=R_\theta(\theta_0) R_\phi(\phi_0)$, where $R_\theta$ and $R_\phi$ are elemental rotations about the quadrotor pitch and roll axes by angles $\theta_0$ and $\phi_0$, respectively. The parameters $x_0$, $y_0$, $z_0$, $\theta_0$, $\phi_0$ are randomized using uniform distributions taking values, respectively, in the intervals $[-5,\, 0]$ m, $[-2.5,\, 2.5]$ m, $[1,\, 6]$ m, $(-\pi,\, \pi)$ rad, $(-\pi,\, \pi)$ rad. The position trajectories followed by the quadrotor over 100 simulation runs, lasting 20 second each, are depicted in Fig.~\ref{dronetraj}. The initial transient of the corresponding position errors is reported in Fig.~\ref{trackerr}, together with the thrust direction error response (evaluated as $\eta=\arccos(\zeta^T x_3)$). It can be seen that the quadrotor asymptotically tracks the reference trajectory for all realizations, even those starting close to the boundary of the stability region (i.e., to the configuration $\eta=\pi$ rad).

\begin{figure}[t]
  \centering
  \includegraphics[width=0.95\columnwidth]{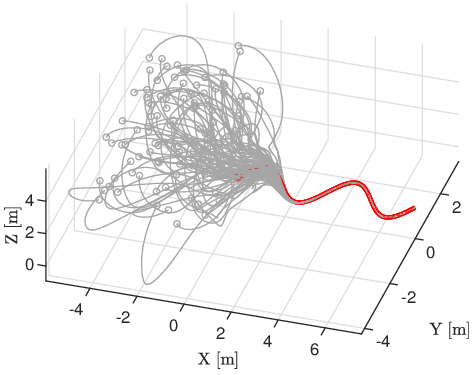}\\
  \caption{Position trajectories followed by the quadrotor over 100 simulation runs (thin gray lines), and reference position trajectory $p_r$ (red line). The initial conditions are marked with a circle.}\label{dronetraj}
\end{figure}
\begin{figure}[t]
  \centering
  \includegraphics[width=0.95\columnwidth]{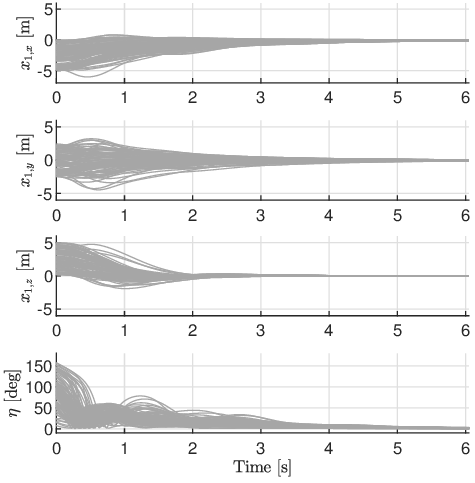}\\
  \caption{Initial transient of the position errors $x_1=[x_{1,x}\,x_{1,y}\,x_{1,z}]^T$ and of the trust direction error $\eta$, corresponding to the trajectories in Fig.~\ref{dronetraj}.}\label{trackerr}
\end{figure}
\begin{figure}[t]
  \centering
  \includegraphics[width=0.95\columnwidth]{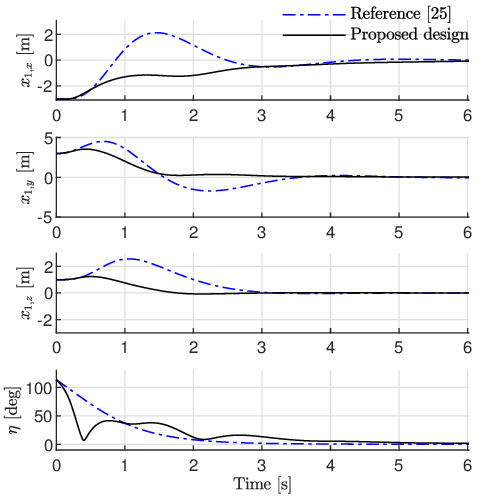}\\
  \caption{Initial transient of the position errors $x_1=[x_{1,x}\,x_{1,y}\,x_{1,z}]^T$ and of the trust direction error $\eta$ for the proposed control law (black solid line) and the control law in \cite{kooijman2019trajectory} (blue dash-dotted line).}\label{trackbsl}
\end{figure}
\begin{figure}[t]
  \centering
  \includegraphics[width=0.95\columnwidth]{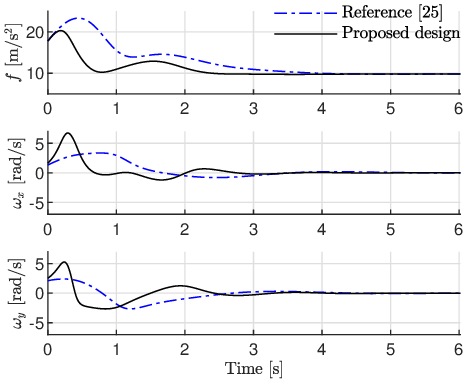}\\
  \caption{Control inputs $f$, $\omega_x$ and $\omega_y$ corresponding to the trajectories in Fig.~\ref{trackbsl}; $\omega_x$ and $\omega_y$ denote the first two entries of $\omega$.}\label{inpbsl}
\end{figure}
\begin{figure}[t]
  \centering
  \includegraphics[width=0.95\columnwidth]{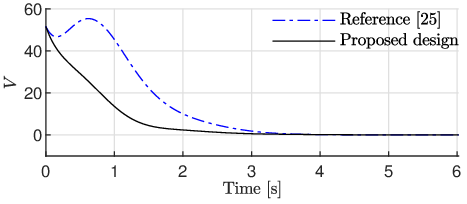}\\
  \caption{Evolution of the Lyapunov function $V$ along the closed-loop system trajectories.}\label{vbsl}
\end{figure}
Next, the proposed controller is compared with that in \cite{kooijman2019trajectory} on an example featuring $x_0=-3$ m, $y_0=3$ m, $z_0=2$ m, $\theta_0=0$ rad, $\phi_0=1$ rad. The thrust direction control law in \cite{kooijman2019trajectory} is recovered by neglecting, in \eqref{claw}, the correction term $\beta$. The tracking errors obtained with the two controllers are depicted in Fig.~\ref{trackbsl}. It can be seen that the proposed design achieves a much better transient performance than the comparison one with regard to position tracking, thanks to a more aggressive initial regulation of the thrust direction vector. This is due to the adaptation mechanism involved in \eqref{betadef}, which is a distinguishing feature of our design. The control inputs returned by the two controllers are depicted in Fig.~\ref{inpbsl}. The proposed control law incurs higher peak values of the angular velocity inputs. This is expected, given the faster transient response. Nevertheless, the overall energy expenditure is significantly smaller, mainly due to a reduced thrust control effort. Figure~\ref{vbsl} shows the evolution of the Lyapunov function $V$. It is confirmed that such a function decreases monotonically along the trajectories produced by the proposed controller. This is not the case for the controller in \cite{kooijman2019trajectory}.

The above results clearly demonstrate the global tracking capability and the effectiveness of our design.

\subsection{Real-World Validation}

Real-world experiments are carried out using a laboratory quadrotor platform, in order to evaluate the impact of unmodeled dynamics and perturbations on our design. The quadrotor platform is characterized by a carbon-fiber frame with a symmetric X shape and a diagonal of $0.2$ m. The vehicle is equipped with a PixRacer Pro flashed with Ardupilot\footnote{\url{https://ardupilot.org/}} for low-level control, and with an nVidia Jetson Xavier NX as an additional computation unit. The total weight of the vehicle is of about 1.0 kg. The Robot Operating System
(ROS) is employed for communication. In particular, we use the package MAVROS \footnote{\url{https://github.com/mavlink/mavros}} to directly send to the low-level controller the thrust and angular velocity commands computed by the proposed control law onboard the Jetson unit. We employ the OptiTrack
motion capture system to obtain the pose information needed for controller feedback. The quadrotor is shown in Figure \ref{fig:drone-setup}, and it is slightly unbalanced due to the presence of the Jetson module.

\begin{figure}[t]
    \centering
    \includegraphics[width=0.95\columnwidth]{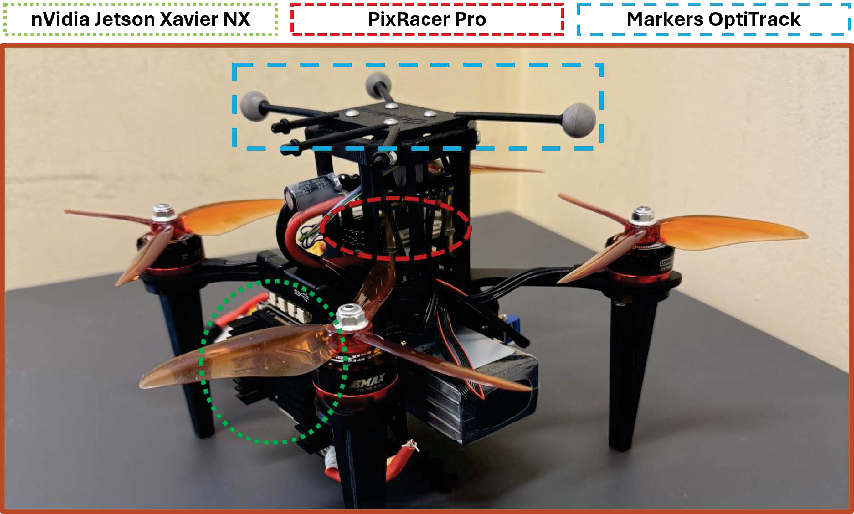}
    \caption{Quadrotor platform employed for real-world validation. The following components are higlighted: (i) nVidia Jetson Xavier NX; (ii) PixRacer Pro; (iii)  OptiTrack markers.}
    \label{fig:drone-setup}
\end{figure}
\begin{figure}[t]
  \centering
  \includegraphics[width=0.95\columnwidth]{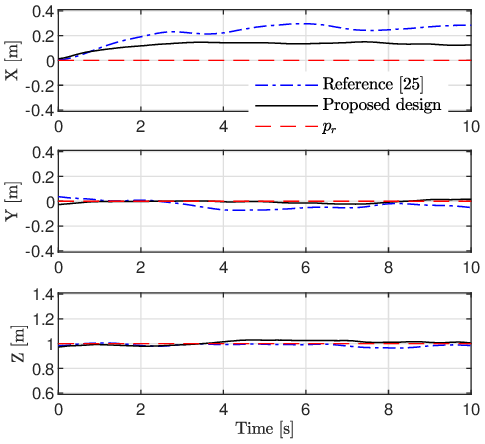}
  \caption{Comparison of the control law in \cite{kooijman2019trajectory} with the proposed one in the hovering test.}\label{exp1}
\end{figure}
\begin{figure}[t]
  \centering
  \includegraphics[width=0.95\columnwidth]{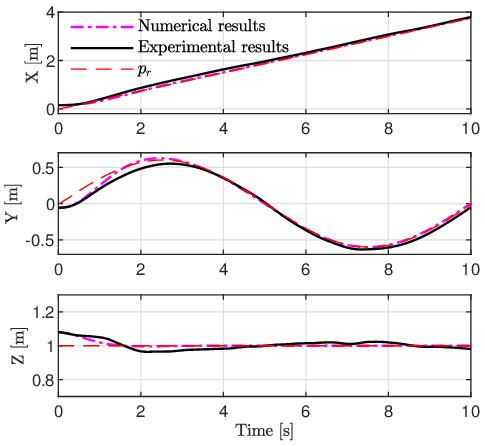}
  \caption{Comparison of numerical and experimental results for trajectory tracking.}\label{exp2}
\end{figure}
\begin{figure}[t]
  \centering
  \includegraphics[width=0.95\columnwidth]{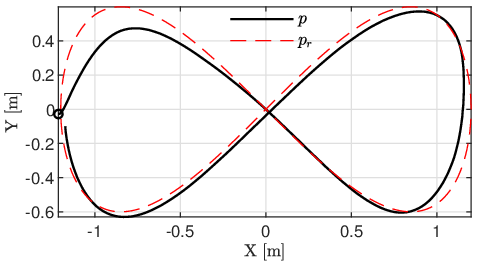}
  \caption{Planar projection of an 8-shaped trajectory that has been tested in a laboratory experiment. In this experiment, the steady-state error visible in Fig. \ref{exp1} is compensated for by adding a constant offset to the first entry of the error vector $x_1$. The initial (hovering) condition is marked with a circle.}\label{exp3}
\end{figure}

As a first experiment, the proposed control law is compared with that in \cite{kooijman2019trajectory} on a simple hovering test.
In particular, the quadrotor is driven towards the initial hovering position $p_r=[0\,0\,1]^T$ m, and each controller is subsequently activated (at $t=0$) to keep this position. The results of this experiment are depicted in Fig. \ref{exp1}. Both controllers are able to maintain hovering flight, but display an offset along the X-axis with respect to the setpoint $p_r$. This is not surprising, in view of the aforementioned quadrotor unbalance and the lack of integral action in  the position control law \eqref{auxcon}. Nevertheless, the steady-state error obtained with the proposed design ($0.15$ m) is half of that produced by the design in \cite{kooijman2019trajectory} ($0.3$ m). This shows that, besides improving the transient performance, the inclusion of the term $\beta$ in \eqref{claw} can significantly reduce the steady-state error due to perturbations. The addition of integral terms to the control law is an interesting topic that will be investigated in future iterations of this work.

A second experiment is done to evaluate the capability of the quadrotor to track the reference trajectory \eqref{rtr}. Figure \ref{exp2} compares the results of this experiment with those obtained numerically. Overall, there is a close matching between numerical and experimental data, indicating that our design is endowed with some intrinsic robustness.

More aggressive maneuvers have also been tested, see for instance the 8-shaped trajectory in Fig. \ref{exp3}, where the proposed controller retains a good tracking performance. 

\section{Conclusion}\label{sec:conclusion}

A new control design methodology has been presented for quadrotor trajectory tracking. The proposed design employs the unit sphere as the configuration manifold for thrust direction control. This allows for the derivation of a high-performance continuous stabilizer featuring almost global stability properties. The resulting control law has been extensively tested both numerically and on real-world experiments, showing consistent performance figures, except possibly for small offsets due to unmodeled perturbations. Future research will tackle the extension of the obtained stability results to the angular velocity subsystem and the inclusion of integral action in the control law for disturbance rejection.

\appendix
\subsection{Derivation of the auxiliary direction vector dynamics}\label{appA1}
The following steps are employed to derive \eqref{uvdyn}-\eqref{omegav}. The time derivative of $v/\|v \|$ can be expressed as
\begin{equation}\label{dnv}
\frac{\text{d}}{\text{d} t}\left( \frac{v}{\|v\|}\right)=\left(I-\frac{v v^T}{\|v \|^2}\right)\frac{\dot{v}}{\|v \|}\,.
\end{equation}
Observing that
\begin{equation}
I-\frac{v v^T}{\|v \|^2}=-\frac{[v]_\times^2}{\|v \|^2}\,,
\end{equation}
one can rewrite \eqref{dnv} as
\begin{equation}\label{dnv2}
\frac{\text{d}}{\text{d} t}\left( \frac{v}{\|v\|}\right)=-\frac{[v]_\times}{\| v \|} \left(\frac{[v]_\times}{\|v \|^2}\,{\dot{v}}\right)=[\omega_v]_\times \frac{v}{\| v \|}\,,
\end{equation}
where $\omega_v$ is given by \eqref{omegav}.
Differentiating \eqref{auv} with respect to time and using \eqref{sysmodel2}, \eqref{auv} and \eqref{dnv2}, one gets
\begin{equation}\label{dx3app}
\dot{x}_3=R [\omega_v]_\times \frac{v}{\| v \|}-[\omega]_\times x_3 =R [\omega_v]_\times R^T x_3-[\omega]_\times x_3 \,.
\end{equation}
Applying the identity $R [\omega_v]_\times R^T=[R \omega_v]_\times$
in \eqref{dx3app} gives~\eqref{uvdyn}.

\balance
\bibliographystyle{IEEEtran}
\bibliography{biblio}

\end{document}